\documentclass[12pt]{iopart}
\usepackage{graphicx}
\usepackage{tikz}
\usetikzlibrary{arrows}
\usepackage{amsthm,amsfonts,amssymb,epsfig,graphics,amsbsy,subfigure}

\usetikzlibrary{arrows,shapes,positioning}
\usetikzlibrary{decorations.markings}
\tikzstyle arrowstyle=[scale=1]
\tikzstyle directed=[postaction={decorate,decoration={markings,
    mark=at position .65 with {\arrow[arrowstyle]{stealth}}}}]
\usepackage{color}		
\usepackage{epsfig}

\usepackage{graphicx}
\newcommand{\intR}{\int\limits_{\mathbb{R}} }
\newcommand{\intRR}{\int\limits_{\mathbb{R}^2} }
\newcommand{\intL}{\int\limits }
\newcommand{\half}{^\infty_0 }
\newcommand{\RR}{\mathbb{R}}

\newcommand{\aaa}{\mathbf{a}}
\newcommand{\xx}{\mathbf{x}}
\newcommand{\yy}{\mathbf{y}}

\newcommand{\ww}{\mathbf{w}}
\newcommand{\oomega}{\boldsymbol{\omega}}

\newcommand{\bbeta}{\boldsymbol{\beta}}
\newcommand{\aalpha}{\boldsymbol{\alpha}}
\newtheorem{thm}{Theorem}
\newtheorem{defi}[thm]{Definition}

\newtheorem{lem}[thm]{Lemma}

\begin{document}


\title{Some uniqueness theorems for a conical Radon transform}

\author[S. Moon]{Sunghwan Moon }
\address{Department of Mathematics\\
 Kyungpook National University\\
 Daegu 41566, Republic of Korea\\
   \ead{ sunghwan.moon@knu.ac.kr}}

\begin{abstract}
The conical Radon transform, which assigns to a given function $f$ on $\RR^3$ its integrals over conical surfaces, arises in several imaging techniques, e.g. in astronomy and homeland security, especially when the so-called Compton cameras are involved.
In many practical situations we know this transform only on a subset of its domain. 
In these situations, it is a natural question what we can say about $f$ from partial information.
In this paper, we investigate some uniqueness theorems regarding a conical Radon transform. 
\vspace{2pc}

 \end{abstract}


 \vspace*{-16pt}


\section{Introduction}
The conical Radon transform is the integral transform that maps a function $f$ on $\RR^3$ to its integrals over conical surfaces. 
This transform is related to various imaging techniques, e.g. in optical imaging \cite{florescums11}, primarily when the so-called \textbf{Compton cameras} are used.
This camera was introduced for use for Single Photon Emission Computed Tomography(SPECT) \cite{singh83} but it is also used in astronomy, and homeland security imaging \cite{allmarasdhkk13}.
More information on Compton cameras can be found, for example, in \cite{allmarasdhkk13,baskozg98,everettftn77,kuchmentt16,moonsiims17,singh83}.

Inversion formulas for various types of conical Radon transforms are provided in~\cite{baskozg98,cebeiromn15,creeb94,gouia14,gouiaa14,haltmeier14,kuchmentt16,maximfp09,moonip15,moonsiims17,moonsima16,moonsiims16,nguyentg05,terzioglu15,smith05,truongnz07}.
However, only a few of these articles study uniqueness of a certain kind of a conical Radon transform (for example, \cite{moonsima16}). 
Here, we investigate the uniqueness property of the conical Radon transform which was first introduced in \cite{moonsiims16}.
This conical Radon transform arises in an application of a Compton camera consisting of two linear detectors.

The conical Radon transform of $f\in C^\infty_c(\RR^3)$ is defined by
\begin{equation}\label{cone_tran}
Cf(u,\beta,s)=\displaystyle\int\limits_{S^2}\intL\half f((u,0,0)+r\boldsymbol\alpha)r\delta(\aalpha\cdot\bbeta-s) {\rm d}r {\rm d}S(\aalpha),
\end{equation}
when $s$ is between $-1$ and 1 and $\bbeta=(\cos\beta,0,\sin\beta)\in S^2$.
(As mentioned in \cite[eq.(2.4)]{moonsiims16}, $s$ is the opening angle of the cone of integration and $\bbeta$ is the unit vector which indicates the central axis of the cone.)
Here $\delta$ is the Dirac delta function and ${\rm d}S(\aalpha)$ is the standard measure on the unit sphere.
When $s$ is outside of the interval $[-1,1]$, $Cf$ is set to be zero.
If $f(\xx),\xx=(x_1,x_2,x_3)\in\RR^3$ is odd with respect to $x_2$, then $Cf$ is equal to zero. 
We thus assume that $f$ is compactly supported in $\{\xx=(x_1,x_2,x_3)\in\RR^3:x_2>0\}$. 

Our results can be applicable to the conical Radon transform obtained by the typical Compton camera, which consists of two planar detectors.
As a formula, this is 
\begin{equation}\label{eq: cone_tran}
\begin{array}{ll}
\displaystyle C_Tf(u_1,u_2,\bbeta,s)=\displaystyle\int\limits_{S^2}\intL\half f((u_1,u_2,0)+r\boldsymbol\alpha)r\delta(\aalpha\cdot\bbeta-s) {\rm d}r {\rm d}S(\aalpha), \\
\qquad\mbox{for}\quad(u_1,u_2,\bbeta,s)\in\RR^2\times S^2\times [-1,1].
\end{array}
\end{equation}
Applicability of our results to this conical Radon transform $C_Tf$ follows from the fact that $C_Tf$ contains $Cf$.

The conical Radon transform can be decomposed into two following transforms:
\begin{defi}\indent
\begin{itemize}
\item The spherical sectional transform $Q$ maps a continuous function $\phi\in C (\RR\times S^2)$ into
$$
Q\phi(u,\beta,s)=\displaystyle\intL_{S^{2}}\phi (u,\boldsymbol\alpha) \delta(\boldsymbol\alpha\cdot(\cos\beta,0,\sin\beta)-s){\rm d}S(\aalpha) 
$$
for $(u,\beta,s)\in\RR\times [0,2\pi)\times[-1,1].$
\item The weighted ray transform $P_{}$ maps a continuous function $f$ on $\RR^3$ with compact support into
$$
P_{}f(u,\ww)=\intL\half f((u,0,0)+r\ww)r^{}{\rm d}r,
$$
for $(u,\ww)=(u,w_1,w_2,w_3)\in \RR\times \RR^3.$
\end{itemize}
\end{defi}
Many articles including \cite{gindikinrs93,hielscherq15,nattererw01,rubin03} studied this spherical sectional transform $Q\phi$.
However only a few articles \cite{hamakerssw80,moonsiims16,moonsiims17,moonjiip161} have investigated $Pf$, although various articles \cite{finchs83,finch85,smith90,smithsw77,solmon76} did study
$\int\half f((u,0,0)+r\ww){\rm d}r$.

The next theorem follows from the definition of $C,Q$ and $P$ and is well known.
\begin{thm}\label{thm:decomposition}
If $f\in C^\infty(\RR^3)$ has compact support in $\{(\xx\in\RR^3:x_2>0\}$, then we have $Cf=Q(Pf)$.
\end{thm}

The next section is devoted to three uniqueness theorems for the conical Radon transform.
\section{Uniqueness theorems}\label{sec:3dimension}
In this section, we present some uniqueness theorems for the conical Radon transform $Cf$.


\begin{thm}\label{thm:unique1}
Let $\mathfrak S\subset S^1$ be a set such that no non-trivial homogeneous polynomial vanishes on $\mathfrak S$.
If $f\in C(\RR^3)$ has compact support in $\{\xx \in\RR^3:x_2>0\}$ and $Cf(u,\beta,s)=0$ for $(\cos\beta,\sin\beta)\in \mathfrak S$ and $(u,s)\in \RR\times[-1,1]$, then $f=0$.
\end{thm}
This theorem is similar to one for the regular Radon transform derived in \cite[Theorem 3.4 in Chapter II]{natterer01}. 
The proof follows the idea in \cite{natterer01}.
\begin{proof}
Let $u\in\RR$ be fixed.
Taking the 1-dimensional Fourier transform of $Cf$ with respect to $s$, we have
$$
\mathcal F_s (Cf)(u,\beta,\sigma)=\intL_{S^2}\intL\half f((u,0,0)+r\aalpha)e^{-\mathrm{i}\sigma\aalpha\cdot(\cos\beta,0,\sin\beta)}r{\rm d}r{\rm d}S(\aalpha).
$$
Since $e^{-\mathrm{i}\sigma\aalpha\cdot(\cos\beta,0,\sin\beta)}=\sum^{\infty}_{j=1}(-\mathrm{i}\sigma\aalpha\cdot(\cos\beta,0,\sin\beta))^j/j!$ and $\mathcal F_s (Cf)(u,\beta,\sigma)=0$ for $(\cos\beta,\sin\beta)\in \mathfrak S$ and $(u,\sigma)\in \RR\times\RR$,
we have
$$
\mathcal F_s (Cf)(u,\beta,\sigma)=\sum^{\infty}_{j=1}a_j(\sigma(\cos\beta,\sin\beta))=0,
$$
where 
$$
\begin{array}{ll}
a_j(\sigma(\cos\beta,\sin\beta))\\
\qquad\displaystyle=\frac{(-\mathrm{i}\sigma)^j}{j!}\intL_{S^2}\intL\half f((u,0,0)+r\aalpha)(\aalpha\cdot(\cos\beta,0,\sin\beta))^jr{\rm d}r{\rm d}S(\aalpha)
\end{array}
$$
for $(\cos\beta,\sin\beta)\in \mathfrak S$ and $(u,\sigma)\in \RR\times\RR$.
Notice that $a_j$ is a homogeneous polynomial of degree $k$.
Since no non-trivial homogeneous polynomial vanishes on $\mathfrak S$, we have for a fixed $u\in\RR$ and any $j=0,1,2,\cdots,$ $a_j=0.$
Therefore, $\mathcal F_s(Cf)$ is equal to zero and by the inversion formula \cite[Theorem 5]{moonsiims16}, $f$ is equal to zero.
\end{proof}
The following theorem is similar to one for the regular divergent beam transform derived in \cite[Theorem 3.5]{natterer01}.
\begin{thm}\label{thm:unique2}
Let $U=\{\xx\in\RR^3:|\xx-(0,0,1)|<1,x_2>0\}$ and let $A\subset\RR$ be infinite.
If $f\in C^\infty(\RR^3)$ has compact support in $U$ and $Cf(u,\beta,s)=0$ for $(u,\beta,s)\in A\times[0,2\pi)\times[-1,1]$, then $f=0$.
\end{thm}
\begin{proof}
By Theorem \ref{thm:decomposition}, $Cf(u,\beta,s)=Q(Pf)(u,\beta,s)$ is equal to zero for $(u,\beta,s)\in A\times[0,2\pi)\times[-1,1]$.
It is known that $Q\phi(u,\cdot,\cdot)$ for a fixed $u$ is injective. (Actually, Gindikin, Reeds, and Shepp derived the inversion formula for $Q$ in \cite[eq. (1.8)]{gindikinrs93} and Natterer and W\"{u}berring mentioned their idea in \cite[page 34]{nattererw01}.) 
Thus, we have $Pf(u,\aalpha)=0$ for $(u,\aalpha)\in A\times S^2$. 
For simplicity, we shift the Cartesian coordinate by $(0,0,1)$.
We will show that if $f\in C^\infty(\RR^3)$ has compact support in $U$ and $Pf(u,\aalpha)=0$ for $(u,\aalpha)\in A\times S^2$, then $f=0$.\footnote{Actually, this argument for the regular divergent transform $Df$ is already known (e.g. \cite[Theorem 5.1 5.6]{hamakerssw80} and \cite[Theorem 5.5]{keinert89}), where
$$
Df(u,\aalpha)=\int\half f((u,0,-1)+r\aalpha){\rm d}r.
$$ 
Also, this is argument very similar to Theorem 3.6 in Chapter II of \cite{natterer01}. Thus the rest of proof is very similar to one of Theorem 3.6.}
Here $U=\{\xx\in\RR^3:|\xx|<1,x_2>0\}$ and 
$$
Pf(u,\aalpha)=\displaystyle\intL\half f((u,0,-1)+r\boldsymbol\alpha)r{\rm d}r .
$$
Let $\oomega_0$ be a limit point of $\{(u,0,-1)/|(u,0,-1)|:u\in A\}$.
Also, let us choose $\epsilon>0$ such that $f(\xx)=0$ for $|\xx|>1-2\epsilon$. 
After (possibly) removing finite elements from $A$ we may assume that for a suitable neighborhood $N(\oomega_0)$ of $\oomega_0$ on $S^2$, 
$$
(u,0,-1)\cdot\oomega>1-\epsilon,\qquad\mbox{for all } u\in A,\quad\oomega\in N(\oomega_0).
$$
Now we consider for $u\in A$ and $\oomega\in N(\oomega_0)$
\begin{equation}\label{eq:pfintegral}
\begin{array}{ll}
 \displaystyle\intL_{S^2}P_{}f(u,\aalpha)\frac{{\rm d}S(\aalpha)}{\aalpha\cdot\oomega}\displaystyle=\intL_{S^{2}}\intL\half f((u,0,-1)+r\aalpha)\frac{r{\rm d}r{\rm d}S(\aalpha)}{\aalpha\cdot\oomega}\\
 \qquad\displaystyle=\intL_{\RR^3} f((u,0,-1)+\yy)\frac{{\rm d}\yy}{\yy\cdot\oomega}\displaystyle=\intL_{|\xx|<1-2\epsilon} \frac{f(\xx){\rm d}\xx}{\xx\cdot\oomega-(u,0,-1)\cdot\oomega},
\end{array}
\end{equation}
where in the second and third equalities, we changed variables $r\aalpha\to\yy$ and $(u,0,-1)+\yy\to\xx$, respectively.
Note that since for $\xx\in U$ with $f(\xx)\neq 0$, $u\in A$ and $\oomega\in N(\oomega_0)$,
$$
\xx\cdot\oomega-(u,0,-1)\cdot\oomega<1-2\epsilon-(1-\epsilon)=-\epsilon,
$$
the last integral in (\ref{eq:pfintegral}) is well defined.
Putting $\xx=s\oomega+\boldsymbol\tau,\boldsymbol\tau\in \oomega^\perp$ in the last integral of (\ref{eq:pfintegral}) gives
$$
\begin{array}{ll}
\displaystyle\intL_{S^2}P_{}f(u,\aalpha)\frac{{\rm d}S(\aalpha)}{\aalpha\cdot\oomega}\displaystyle=\intL_{\RR}\intL_{\oomega^\perp} f(s\oomega+\boldsymbol\tau)\frac{{\rm d}\boldsymbol\tau {\rm d}s}{s-(u,0,-1)\cdot\oomega}\\
\quad\displaystyle=\intL_{\RR} Rf(\oomega,s)\frac{ {\rm d}s}{s-(u,0,-1)\cdot\oomega},
\end{array}
$$
where $Rf$ is the 3-dimensional regular Radon transform, i.e.,
$$
Rf(\oomega,s)=\intL_{\RR^3}f(\xx)\delta(\xx\cdot\oomega-s){\rm d}\xx.
$$
Since $Pf(u,\aalpha)=0$ for $(u,\aalpha)\in A\times S^2$, we have
\begin{equation}\label{eq:intrf}
 \intL_{|s|<1-2\epsilon}Rf(\oomega,s)\frac{ {\rm d}s}{s-(u,0,-1)\cdot\oomega}=\displaystyle\intL_{S^2}P_{ }f(u,\aalpha)\frac{{\rm d}S(\aalpha)}{\aalpha\cdot\oomega}=0
\end{equation}
for $u\in A$ and $\oomega\in N(\oomega_0).$
Let us consider the power series 
$$
(s-(u,0,-1)\cdot\oomega)^{-1}=-\sum^\infty_{j=0}s^j((u,0,-1)\cdot\oomega)^{-1-j},
$$
which converges uniformly in $|s|\leq 1-2\epsilon$ for any $u\in A$ and $\oomega\in N(\oomega_0)$, since $|s|< 1-\epsilon<(u,0,-1)\cdot\oomega$.
Together with (\ref{eq:intrf}), we have
\begin{equation}\label{eq:intrf1}
\begin{array}{rl}
  0&\displaystyle=\intL_{S^2}P_{ }f(u,\aalpha)\frac{{\rm d}S(\aalpha)}{\aalpha\cdot\oomega}=-\sum^\infty_{j=0}((u,0,-1)\cdot\oomega)^{-1-j}\intL_{|s|<1-2\epsilon}Rf(\oomega,s)s^j {\rm d}s\\
&\displaystyle=-\sum^\infty_{j=0}((u,0,-1)\cdot\oomega)^{-1-j}p_j(\oomega),
\end{array}
\end{equation}
where 
$$
p_j(\oomega)=\intR Rf(\oomega,s)s^j{\rm d}s.
$$
By the range description of the regular Radon transform (for example, see \cite[Theorem 4.1 in Chapter 2]{natterer01}), $p_j$ is a homogeneous polynomial of degree $j$.
Since $(u,0,-1)\cdot\oomega>1-\epsilon$ for $\oomega\in  N(\oomega_0)$ and $u\in A$, $\{((u,0,-1)\cdot\oomega)^{-1}:u\in A\}$ is bounded.
For $\oomega\in  N(\oomega_0)\backslash\{(0,0,-1)\}$, $\{((u,0,-1)\cdot\oomega)^{-1}:u\in A\}$ has a limit point, since $\{(u,0,-1)\cdot\oomega:u\in A\}$ is infinite.
The power series which vanishes on a set with a limit point is identically zero, so varying $u\in A$ gives for $\oomega\in  N(\oomega_0)$ and each $j$,
$$
0=p_j(\oomega)=\intR Rf(\oomega,s)s^j{\rm d}s.
$$
Since the set of polynomials is dense in $L^2(-1,1)$, we conclude $Rf(\oomega,s)=0$ for $(\oomega,s)\in  N(\oomega_0)\times [-1+2\epsilon, 1-2\epsilon]$, hence $f=0$ since $\mathcal Ff(s\oomega)=\mathcal F_s(Rf)(\oomega,\sigma)$ is harmonic and $\mathcal Ff(s\oomega)=0$ for any $\oomega\in N(\oomega_0)$ containing the limit point.
\end{proof}
We introduce two notations.
For any set $U\subset\RR^3$, the closure of $U$ is denoted by $cl(U)$.
For an open set $V\subset\RR^3$, we write $U\subset\subset V$, if $U\subset cl(U)\subset V$ and $cl(U)$ is compact.

\begin{thm}\label{thm:unique3}
Let $U\subset\subset \{\xx\in\RR^3:x_2>0\}$ be a convex and open set and $A\subset\RR$ be a connected set.
Also, let $f\in C^\infty(\RR^3)$ have compact support in $U$.
Suppose for each $u\in A$, $S'(u)=\{(\alpha_1,\alpha_3)\in \RR^2:(\alpha_1,\sqrt{1-|(\alpha_1,\alpha_3)|^2},\alpha_3)\in S(u)\}$ is convex, where
$S(u)=\{\aalpha\in S^2: \{((u,0,0)+r\aalpha)\in\RR^3:r\in[0,\infty)\}\cap cl (U)\neq\emptyset\}$.
If $Cf(u,\beta,s)=0$ for every cone $\{((u,0,0)+r\aalpha)\in\RR^3:u\in A,r\in[0,\infty),\aalpha\cdot(\cos\beta,0,\sin\beta)=s\}$ not meeting $cl(U)$, then $f=0$ outside $W$, where
$$\begin{array}{ll}
W=\{(u,0,0)+r\aalpha:u\in A,r\geq0,\aalpha\in \cup_{v\in A}S(v)\}.
\end{array}
$$
\end{thm}
\begin{proof}
For the moment, let $u\in \RR$ be fixed.
For 
$$
\Phi(u,\yy)=\left\{\begin{array}{ll}\displaystyle\frac{\phi(u,y_1,\sqrt{1-|\yy|^2},y_2)}{\sqrt{1-|\yy|^2}}\quad&\mbox{for }|\yy|<1,\yy=(y_1,y_2)\in\RR^2,\\
0&\mbox{otherwise,}\end{array}\right.
$$
we have $Q\phi(u,\beta,s)=R_2\Phi(u,\beta,s)$, where $R_2$ is the 2-dimensional regular Radon transform, i.e.,
$$
R_2\Phi(u,\beta,s)=\intRR\Phi(u,\yy)\delta(\yy\cdot(\cos\beta,\sin\beta)-s){\rm d}\yy.
$$
From Theorem \ref{thm:decomposition}, we have $Cf=R_2\Phi$ where
$$
\Phi(u,\yy)=\left\{\begin{array}{ll}\displaystyle\frac{Pf(u,y_1,\sqrt{1-|\yy|^2},y_2)}{\sqrt{1-|\yy|^2}}\quad&\mbox{for }|\yy|<1,\\
0&\mbox{otherwise.}\end{array}\right.
$$
By the hole theorem \cite[Theorem 3.1 in Chapter II]{natterer01} for the regular Radon transform, $\Phi(u,\cdot)$ is supported in $S'(u)$, since $S'(u)$ is a convex and compact set by the compactness of $S(u)$, and thus $Pf(u,\aalpha)=0$ for $\{(u,0,0)+r\aalpha:r\geq0,\aalpha\in S(u)\}$ not meeting $U$.  
The following lemma completes our proof.

\begin{lem}
Let $S$ be an open set on $S^{2}$, and let $\mathfrak A$ be a continuously differentiable curve outside $U$. Let $U\subset\RR^3$ be bounded and open.
Assume that for each $\aalpha\in S$ there is an $\aaa \in \mathfrak A$ such that the half-line $\aaa +r\aalpha,r\geq0$ misses $U$.
If $f\in C^\infty(\RR^3)$ has compact support in $U$ and $\int\half f(\aaa +r\aalpha)r{\rm d}r=0$ for $\aaa \in \mathfrak A$ and $\aalpha\in S$, then $f=0$ in $\{\aaa +r\aalpha:\aaa\in \mathfrak A,\aalpha\in S,r\geq0\}$.
\end{lem}
This lemma is very similar to Theorem 3.3 in Chapter II \cite{natterer01}. 
\begin{proof}
Similar to the proof of Theorem 3.3 in Chapter II \cite{natterer01}, we can show 
$$
\intL\half r^kf(\aaa +r\aalpha){\rm d}r=0\qquad\mbox{for }(\aaa,\aalpha)\in \mathfrak A\times S\mbox{ and } k\geq1.
$$
Note that by the Stone-Weierstrass theorem, the space spanned by $\{r^k:k\geq1\}$ is dense in the set 
$$
\{g\in C^\infty(\RR):g(0)=0,g\mbox{ has compact support}\}.
$$
Thus for a fixed $(\aaa,\aalpha)\in \mathfrak A\times S$, we have $f(\aaa+r\aalpha)=0$ on $r\in [0,\infty)$. 
Since $f$ has compact support in $U$ and $\mathfrak A$ is outside $U$, $f=0$ in $\{\aaa+r\aalpha:\aaa\in \mathfrak A,\aalpha\in S,r\geq0\}$.
\end{proof}
\end{proof}
%
\section*{Acknowledgments}
This work was supported by the National Research Foundation of Korea grant funded by the Korea government (MSIP) (2015R1C1A1A01051674).
\section*{References}
\bibliographystyle{sapm}


\end{document}